%% file: main.tex
\documentclass[11pt]{article}

\usepackage{cite}
\usepackage{amsmath,amssymb,amsfonts}
\usepackage{fullpage}
\usepackage{graphicx}
\usepackage{textcomp}
\usepackage{xcolor}
\usepackage{hyperref}

\usepackage{amsthm}
\usepackage{algorithm}
\usepackage[noend]{algpseudocode}
\usepackage{tikz}
\usetikzlibrary{arrows,shapes.misc,shapes.geometric,shapes.gates.logic.US,shapes.gates.logic.IEC,calc}
\tikzstyle{branch}=[fill,shape=circle,minimum size=3pt,inner sep=0pt]
\usepackage{mathtools}

\input{macros}

\begin{document}

\title{Payment-failure times for random Lightning paths}

\date{}

\makeatletter
\renewcommand{\thefootnote}{\fnsymbol{footnote}} 
\makeatother
\author{
  \begin{tabular}{c@{\extracolsep{8em}}c}
    Taki E.M. Abedesselam\footnotemark[1] & Fabio Giacomelli\footnotemark[1]\\
    \texttt{abedesselam@ing.uniroma2.it} & \texttt{fabio.giacomelli@uniroma2.it}\\[1ex]
    Francesco Pasquale\footnotemark[2] & Michele Salvi\footnotemark[2]\\
    \texttt{pasquale@mat.uniroma2.it} & \texttt{salvi@mat.uniroma2.it}\\[2ex]
  \end{tabular}
}
\footnotetext[1]{University of Rome ``Tor Vergata'' and University of Camerino}
\footnotetext[2]{University of Rome ``Tor Vergata''}

\maketitle

\makeatletter
\renewcommand{\thefootnote}{\arabic{footnote}} 
\makeatother

\setcounter{footnote}{0} 

\begin{abstract}
\input{./abstract}

\end{abstract}

\input{./trunk/intro}

\input{./trunk/model}

\input{./trunk/clique}

\input{./trunk/betweenness}
\input{./trunk/simulations}

\input{./trunk/concl}

\bibliographystyle{plain}
\bibliography{mls}

\appendix
\section*{Appendix}

\input{./trunk/apx}

\end{document}

%% file: macros.tex
\newcommand{\bigO}{\mathcal{O}}
\newcommand{\calP}{\mathcal{P}}
\newcommand{\calE}{\mathcal{E}}

\newcommand{\Prob}[2]{\mathbf{P}_{#1} \left( #2 \right)}
\newcommand{\Expec}[2]{\mathbf{E}_{#1} \left[ #2 \right]}

\newcommand{\bmin}[1]{b_{\min}\left( #1 \right)}
\newcommand{\tx}[4]{\texttt{pay}_{#4} \left( #1, #2, #3 \right)}

\newtheorem{definition}{Definition}[section]
\newtheorem{lemma}[definition]{Lemma}
\newtheorem{theorem}[definition]{Theorem}

%% file: abstract.tex
We study a random process over graphs inspired by the way payments are executed
in the Lightning Network, the main layer-two solution on top of Bitcoin. We
first prove almost tight upper and lower bounds on the time it takes for a
payment failure to occur, as a function of the number of nodes and the edge
capacities, when the underlying graph is complete. Then, we show how such a
random process is related to the edge-betweenness centrality measure and we
prove upper and lower bounds for arbitrary graphs as a function of
edge-betweenness and capacity. Finally, we validate our theoretical results by
running extensive simulations over some classes of graphs, including snapshots
of the real Lightning Network.

%% file: trunk/intro.tex
\section{Introduction}\label{sec:intro}
The Lightning Network~\cite{poon2016bitcoin} is the main layer-two solution on
top of Bitcoin~\cite{nakamoto2008bitcoin} that promises to address the
scalability issue (see, e.g.,~\cite{croman2016scaling}) executing
\textit{off-chain} the vast majority of transactions, and using the blockchain
layer as a notary service to resolve controversies. It consists of a series of
cryptographic mechanisms that allow parties to build a \textit{channel graph}
and to execute \textit{payments} over paths on such a graph in a trustless way,
relying on the security of the underlying blockchain. Roughly speaking, the
channel graph is built as follows: Two parties, $u$ and $v$, can create a
\textit{channel} by locking bitcoins in a 2-of-2 multisignature address and
recording it on the blockchain, the amount of locked funds is the channel
\textit{capacity} and it is known to all the parties in the system; at any
point in time, each one of the two endpoints of a channel owns a share of such
capacity, those shares are called the channel \textit{balance}; the balance of
a channel is known only by the two endpoints and it can be updated by
exchanging private messages between them. For example, if $u$ and $v$ opened a
channel of capacity $5$, and the current balance of the channel is $4$ for $u$
and $1$ for $v$, when $u$ wants to pay $1$ to $v$, $u$ and $v$ can agree to
update the balance of their channel to $3$ for $u$ and $2$ for $v$. The parties
in the system and the channels between them constitute nodes and edges,
respectively, of the \textit{channel graph} of the Lightning Network. A node
$u$ that needs to pay a non-neighbor node $w$ can look for a path $\calP$
between $u$ and $w$ on the channel graph and request the execution of the
payment on that path; in order for the payment to be successful, the balances
of all the channels in path $\calP$ have to be updated accordingly, hence a
payment can \textit{fail} if this is not possible; in this case, we say that a
\textit{payment failure} occurs, and node $u$ has to try another path to pay
$w$. We give a slightly more rigorous description of such a process in
Section~\ref{sec:model} and we refer the reader interested in the details of
the cryptographic mechanisms that allow the routing of a payment to be
\textit{atomic} (either the balance of all the paths is updated or none is)
to~\cite{antonopoulos2022mastering}. 

\subsection{Our contribution}\label{ssec:intro:contrib}
Motivated by the above scenario, in this paper we study the following random
process: Given an undirected graph $G$, where each edge $e$ has a capacity
$c(e)$ and an initial balance equally divided between the two endpoints, at
every discrete round a source node $u$ and a destination node $v$ are chosen
uniformly at random (u.a.r.) among all the nodes, then a shortest path $\calP$
is chosen u.a.r. among all shortest paths between $u$ and $v$, and a
payment of unit value is executed from $u$ to $v$ over path $\calP$.  Our
goal is to investigate how long it takes for a payment failure to occur,
depending on the topology of the graph and on the channel capacities.  

We first analyze the case where the graph $G$ is a clique where every edge has
the same capacity $2k$. We prove that, in this case, the time until the first
failure occurs is $\Omega(k^2 n^2 / \log n)$ and $\bigO(k^2 n^2)$, with high
probability (w.h.p.).

Then, we extend the analysis to general graphs and show that the time until the
first failure occurs depends on the ratio between the squared capacity of the
edges and their \textit{edge-betweenness} \cite{girvan2002community}. More
precisely, in this case we prove that the time until the first failure occurs is
$\Omega\left(\xi \, \frac{n^2}{\log n}\right)$ and $\bigO\left(\xi \, n^2 \log
n\right)$, where $\xi$ is the minimum, over all edges $e$, of the ratio between
the squared capacity $k_e^2$ and the edge-betweenness $g(e)$.

Finally, we validate our theoretical results through extensive simulations of
the random process. We observe that the results of the simulations on the
clique graph are consistent with the theoretical bounds, and that the lower
bound is likely tight. To highlight the impact of the correlation between
edges, we simulate the process on ring graphs with $n$ nodes and we compare the
results with those that we get by considering $n$ independent birth-and-death
chains, each representing the evolution of the capacity of a single edge.
Additionally, we run simulations on a snapshot of the channel graph of the
Lightning Network. Specifically, we investigate how the overall payment failure
rate in the network would improve, if it was possible to rearrange the capacity
of the edges to maximize $\xi$.

\input{./trunk/related}

\subsection{Organization of the paper}
In Section~\ref{sec:model}, we  provide the formal description of the random
process we study in this paper.
In Section~\ref{sec:clique}, we analyze the case where the graph is a clique
with uniform capacities and derive high-probability bounds on the failure time.
In Section~\ref{sec:betweenness} we extend the analysis to general graphs and
show that the failure time depends on the smallest ratio between squared
capacity and edge betweenness. In Section~\ref{sec:simulations}, we present the
simulation results on cliques, rings, and a snapshot of the channel graph of
the Lightning Network graph. Finally, in Section~\ref{sec:conclusions} we draw
some conclusions.

%% file: trunk/related.tex
\subsection{Related work}\label{ssec:related} The Lightning Network, first
described in~\cite{poon2016bitcoin}, is the source of several research problems
related to distributed network formation, and the notion of edge-betweenness
centrality~\cite{girvan2002community} is recurrent in most of them. For
example, in~\cite{lange2021impact} the authors look at different attachment
strategies and show how they affect the network in the long term. One of the
strategies involves adding a node using a greedy algorithm that tries to
maximize its betweenness centrality. This strategy was proposed in~\cite{MBI},
where $k$ edges are added to maximize the node’s betweenness. Alternatively, a
node can be connected to the nodes with the highest betweenness in the network.
Both strategies turned out effective in reducing the network’s diameter and
increasing the success rate of payments in the Lightning Network.
In~\cite{ersoy2020profit}, a greedy algorithm was designed to increase a node's
profit by improving its betweenness centrality.  In~\cite{Rene2021Security}, a
probabilistic model was proposed to account for uncertainty in channel balance
and to support multi-part payments—splitting a payment into parts and sending
them instead of sending it as a whole. While the proposed model does not
consider the dynamics of channel balances, it allows for specifying
service-level objectives and quantifying the amount of private information
leaked to the sender as a side effect of payment attempts.  A survey on
networking problems related to blockchain and cryptocurrencies can be found
in~\cite{dotan2021survey}.

The model we study in this paper is inspired by the Lightning Network, however
similar models have been previously studied in the context of \textit{credit
networks}~\cite{Ghosh2007Mechanism,Defigueiredo2005trustdavis} as solution
concepts for certain types of auctions. What distinguishes the design of the
Lightning Network is that it relies on cryptography to establish channels
between nodes and to allow parties to execute off-chain transactions in a
trustless way. There have also been attempts to understand this type of networks
mathematically, e.g., in~\cite{Dandekar2011liquidity} the authors study the
long-term behavior of a similar random process for rings, stars, lines, cliques,
Erdős–Rényi, and Barabási–Albert topologies. In~\cite{Goel2020ContinuousCN} the
authors provide tools for analyzing credit networks by reformulating
transactions from discrete to continuous transaction models.

In our simulations we need to repeatedly sample shortest paths between nodes,
u.a.r. among all shortest paths. In \cite{OUSPS}, several efficient algorithms
for uniformly sampling shortest paths between fixed source and target nodes
were proposed and analyzed.

%% file: trunk/model.tex
\section{The model and the problem}\label{sec:model}
Consider an edge-weighted undirected graph $G = (V,E)$. In the context of the
Lightning Network~\cite{poon2016bitcoin,antonopoulos2022mastering}, edges are
called \textit{channels} and edge weights $c\,:\, E \rightarrow \mathbb{N}$ are
called \textit{capacities}; throughout this paper we align with that
terminology. The capacity $c(e)$ of a channel $e = \{u,v\}$ is ``shared''
between the two endpoints $u,v$: Such a share, $\{b_e(u), b_e(v)\}$ where
$b_e(u), b_e(v) \in \left\{0, 1, \dots, c(e) \right\}$ and $b_e(u) + b_e(v) =
c(e)$, is called the channel \textit{balance}. For convenience sake, for a
channel $e = \{u,v\}$ we also define $\bmin{e} = \min \{b_e(u), b_e(v)\}$ and
we say that a channel $e$ is \textit{empty in one direction} if $\bmin{e} = 0$.

A \textit{payment} between two adjacent nodes $u$ and $v$ is an update of the
balance of edge $e = \{u,v\}$, e.g., if the current balance between $u$ and $v$
is $\{b_e(u), b_e(v)\}$ and node $u$ executes a payment $\tx{u}{v}{x}{}$ to
node $v$ of an \textit{amount} $x \in \{1, \dots, b_e(u) \}$, then the new
balance $\{b_e'(u), b_e'(v)\}$ of the edge will be $b_e'(u) = b_e(u) - x$ and
$b_e'(v) = b_e(v) + x$.

Payments can be \textit{routed} across the network: A node $u$ willing to send
an amount $x$ to a non-neighbor node $v$ can look for a path $\calP = (u = u_0,
u_1, \dots, u_h = v)$ in the graph between nodes $u$ and $v$, if it exists, and
execute a payment $\tx{u}{v}{x}{\calP}$ to node $v$ over path $\calP$ of an
amount $x$, provided that $x \leqslant b_{e_i}(u_i)$ for every channel $e_i =
\{ u_i, u_{i+1} \}$ in the path. Such a payment will update the balances of all
the channels in $\calP$ accordingly, i.e., for every channel $e_i = \{u_i,
u_{i+1}\}$ in the path the new balance $\left\{ b_{e_i}'(u_i),
b_{e_i}'(u_{i+1})\right\}$ after the payment will be\footnote{In the real
Lightning Network, for each intermediate channel there will be a small
\textit{fee} that the node forwarding the payment will subtract to the amount,
hence if $v$ needs to receive an amount $x$ then node $u$ has to send a larger
amount, say $x + \varepsilon$, where $\varepsilon$ is the total fee that will
be subtracted by the intermediate nodes, $u_1, \dots, u_{k-1}$ in the path. In
order to keep our model simple and to focus it on the main graph theoretic
problem that we want to address, we here ignore such a detail (as well as
several other details) of the actual payment routing process in the Lightning
network.} 
\[
\left\{
\begin{array}{ccl}
b_{e_i}'(u_i) & = & b_{e_i}(u_i) - x \\[2mm]
b_{e_i}'(u_{i + 1}) & = & b_{e_i}(u_{i+1}) + x
\end{array}
\right.
\]

While the graph $G = (V,E)$ and the channel capacities $c\,:\, E \rightarrow
\mathbb{N}$ are known to all the nodes in the network, the balance $\{b_e(u),
b_e(v)\}$ of a channel $e = \{u,v\}$ is known only by the endpoints $u$ and
$v$. Hence, it is possible (and it is often the case in the real Lightning
Network) that when a node $u$ tries to pay an amount $x$ to a non-neighbor node
$v$ using a path $\calP = (u = u_0, u_1, \dots, u_k = v)$, the payment cannot
be executed, due to the fact that for some channel in the path the balance is
not sufficient to accommodate the payment; this happens when there is a channel
$e_i = \left\{u_i, u_{i+1}\right\}$ in the path with balance $\left\{
b_{e_i}(u_i), b_{e_i}(u_{i+1}) \right\}$ such that $b_{e_i}(u_i) < x$. In such
a case we say that a \textit{payment failure} occurs, and node $u$ has to
choose another path to pay node $v$.

\paragraph{The random process.} To formulate a concrete mathematical problem
and theoretically analyze the impact of network topology and channel capacities
on the rate of payment failures, we here consider the following discrete-time
random process. Given an undirected connected graph $G = (V,E)$, where all
channels have the same capacity that is initially perfectly balanced between
the two endpoints, i.e., for each channel $e = \{u, v\} \in E$, we have $c(e) =
2k$ for some $k \in \mathbb{N}$ and we start with $b_e(u) = b_e(v) = k$. At
every round $t = 1,2, \dots$, we pick a source node $u \in V$ uniformly at
random (u.a.r.), a destination node $v \in V$ u.a.r., and a shortest path
$\calP$ between $u$ and $v$, u.a.r. among all the shortest paths between $u$
and $v$, and we execute a payment of amount $x = 1$ over path $\calP$ from $u$
to $v$. We are interested in the expected number of rounds before we have that
a channel $e = \{u,v\}$ exists such that $b_e(u) = 0$ or $b_e(v) = 0$.

\begin{algorithm}
\caption{The random process}\label{alg:random_process}
\begin{algorithmic}[1]
\Require An undirected connected graph $G = (V,E)$; 

\phantom{ai} Edge capacities $c(e) = 2k$ for every $e \in E$; 

\phantom{ai} Initial balance $b_e(u) = b_e(v) = k$, for
every edge $e = \{u,v\}$.

\While{$\bmin{e} > 0$ for every $e \in E$}
\State Pick a source node $u$ in $V$, u.a.r.
\State Pick a destination node $v$ in $V$, u.a.r.
\State Pick a shortest path $\calP = (u = u_0, u_1, \dots, u_h = v)$, 

\phantom{Pick} u.a.r. among all shortest paths between
$u$ and $v$;
\For{$i = 0, \dots, h-1$}
\State $b_{\{u_i, u_{i+1}\}}(u_i) = b_{\{u_i, u_{i+1}\}}(u_i) - 1$
\State $b_{\{u_i, u_{i+1}\}}(u_{i+1}) = b_{\{u_i, u_{i+1}\}}(u_{i+1}) + 1$
\EndFor
\EndWhile
\State \textbf{return} The number of iterations of the while cycle
\end{algorithmic}
\end{algorithm}

%% file: trunk/clique.tex
\section{Complete graph with constant capacities}\label{sec:clique}
In this section we analyze the number of iterations of the while cycle in
Algorithm~\ref{alg:random_process} as a function of the number of nodes and of
the initial capacity of the channels, when the underlying graph $G$ is a
clique. We first prove an equivalence between the random process in
Algorithm~\ref{alg:random_process} and the first hitting time of the boundary
for a family of $m = |E|$ unbiased birth-and-death chains (see the random
process in Algorithm~\ref{alg:multiple_bdc_process}).

\begin{algorithm}[h!]
\caption{Multiple Birth-and-Death chains process}\label{alg:multiple_bdc_process}
\begin{algorithmic}[1]
\Require Two integers $m, k \in \mathbb{N}$
\ForAll{$e = 1, \dots, m$}
\State Set $X(e) = 0$
\EndFor
\While{for every $e$, $X{(e)} \neq \pm k$}
\State Pick $e$, u.a.r.
\State $X(e) = 
\left\{ 
\begin{array}{l}
X(e) + 1 \quad \text{with probability } 1/2 \\
X(e) - 1 \quad \text{with probability } 1/2
\end{array}
\right.$
\EndWhile
\State \textbf{return} The number of iterations of the while cycle
\end{algorithmic}
\end{algorithm}

\begin{lemma}[Equivalence lemma]\label{lemma:equivalence}
Let $\tau_1$ and $\tau_2$ be the random variables indicating the number of
iterations of the while cycles in Algorithms~\ref{alg:random_process}
and~\ref{alg:multiple_bdc_process}, respectively.  Then, when $m$ in
Algorithm~\ref{alg:multiple_bdc_process} equals $|E|$ in
Algorithms~\ref{alg:random_process}, a coupling between the random processes
exists such that $\tau_1 = \tau_2$.
\end{lemma}
\begin{proof}
In a complete graph the edge connecting two nodes is the unique shortest path
between them. Hence, when we pick a random source, a random destination, and a
random shortest path at lines $2-4$ in Algorithm~\ref{alg:random_process}, we
are picking a single edge u.a.r. and a random direction on that edge. Moreover,
for an edge $e = \{u,v\}$, if we fix an arbitrary orientation of the edge and
consider the quantity
\begin{equation}\label{eq:edgebdchain}
X_t(u,v) = \frac{b_e^{(t)}(v) - b_{e}^{(t)}(u)}{2}
\end{equation}

We have that, at each round $t$, such quantity changes in one of three ways: it
increases by one, $X_{t+1}(u,v) = X_t(u,v) + 1$ (if $u$ is picked at line $2$
and $v$ is picked at line $3$ in~Algorithm~\ref{alg:random_process}), it
decreases by one, $X_{t+1}(u,v) = X_t(u,v) - 1$ (if $v$ is picked at line $2$
and $u$ at line $3$) or it remains the same $X_{t+1}(u,v) = X_t(u,v)$ (if the
pair $\{u,v\}$ is not picked at lines $2,3$). Hence, given the random process
defined in Algorithm~\ref{alg:random_process} we can define the random process
in Algorithm~\ref{alg:multiple_bdc_process}, with $m = |E|$, as follows. Let
$\hat{E}$ be an arbitrary orientation of the edges in $E$ (i.e., for every
$\{u,v\} \in E$ either $(u,v) \in \hat{E}$ or $(v,u) \in \hat{E}$) and let
$f\,:\, E \rightarrow [m]$ be an arbitrary bijective map. When in the random
process in Algorithm~\ref{alg:random_process} we pick a source $u$ and a
destination $v$ at lines $2$ and $3$, in the random process in
Algorithm~\ref{alg:multiple_bdc_process} we pick $f\left(\{u,v\}\right)$ at
line~$4$, and at line~$5$ we set $X(e) = X(e) + 1$ if $(u,v) \in \hat{E}$ and
$X(e) = X(e) - 1$ if $(v,u) \in \hat{E}$. By construction, this defines a
coupling of Algorithm~\ref{alg:random_process} and
Algorithm~\ref{alg:multiple_bdc_process} with $\tau_1 = \tau_2$.
\end{proof}

\begin{theorem}\label{thm:clique}
Let $G = (V,E)$ be a complete graph with $n$ nodes. For every channel $e \in
E$, let $c(e) = 2k$ be its capacity, with $k \in \mathbb{N}$ and $k > \sqrt{4
\alpha \log n}$ for some constant $\alpha > 1$, and let $b_e^{(0)}(u) =
b_e^{(0)}(v) = k$ be its initial balance. Let $\tau$ be the random variable
indicating the first time one of the channels is empty in one direction, i.e.,
$\tau = \inf\{t \in \mathbb{N} \,:\, \exists e \in E \text{ with } b_{\min}(e)
= 0 \}$. Then, 

\begin{enumerate}
\item $\tau = \bigO\left(k^2 n^2\right)$, w.h.p.\footnote{We say
that an event $\calE_n$ depending on a parameter $n$, that here indicates the
number of nodes in the graph, holds \textit{with high probability (w.h.p.)} if
a constant $c > 0$ exists such that $\Prob{}{\calE_n} \geqslant 1 - n^{-c}$,
for every sufficiently large $n$.}
\item $\tau = \Omega\left(k^2 n^2 / \log n\right)$, w.h.p.
\end{enumerate}
\end{theorem}

\begin{proof}
Due to the equivalence proved in Lemma~\ref{lemma:equivalence} we can study the
distribution of the number of iterations of the while cycle in
Algorithm~\ref{alg:multiple_bdc_process}: let $\tau$ be the random variable
indicating such number of iterations. Observe that $\tau = \min\{ \tau^{(e)}
\,:\, e = 1, \dots, m \}$ where $\tau^{(e)} = \inf\{t \in \mathbb{N} \,:\,
X_t^{(e)} = \pm k \}$ and here we indicate with $X_t^{(e)}$ the value that
variable $X(e)$ at line $5$ of Algorithm~\ref{alg:multiple_bdc_process} has at
the $t$-th iteration of the while cycle.

For $e = 1, \dots, m$ and for $t = 1, 2, \dots$, let $Y_t(e)$ be the indicator
random variable of the event ``Edge $e$ is chosen at round $t$'' in
Algorithm\ref{alg:multiple_bdc_process}, and let 
\[
\overline{Y}_t(e) = \sum_{i=1}^t Y_i(e)
\]
Since $\Prob{}{Y_i(e)} = 1/m$ for every $e$ and every $i$, it holds that
$\Expec{}{\overline{Y}_t(e)} = t / m$, and observe that for every fixed $e$
random variables $\{Y_i(e) \,:\, i = 1, 2, \dots  \}$ are independent.

\smallskip As for the upper bound, from Chernoff Bound (see
Lemma~\ref{apx:chernoff} in the Appendix~\ref{sec:apx:concentration}) it
follows that, for every $e = 1, \dots, m$
\[
\Prob{}{\overline{Y}_t(e) \leqslant \frac{t}{2m}} \leqslant e^{-\frac{t}{8m}}\,.
\]
Hence, if we chose $t = 4mk^2$, we have that for every $e = 1, \dots, m$, the
probability that $e$ has been chosen less than $2 k^2$ times is at most
\[
\Prob{}{\overline{Y}_t(e) \leqslant 2 k^2} \leqslant e^{-k^2 / 2}\,.
\]
and by using the union bound, the probability that an $e$ exists that has been
chosen less than $2 k^2$ times is
\begin{equation}\label{eq:badevent}
\Prob{}{\exists e \,:\, \overline{Y}_t(e) \leqslant 2 k^2} \leqslant m
e^{-k^2 / 2}\,.
\end{equation}
Thus, if $k > \sqrt{2 \alpha \log n}$ for some $\alpha > 1$, we have that with
probability at least $1-m^{-\alpha+1}$ at round $t = 4mk^2$ all $e \in \{1, \dots,
m\}$ have been chosen at least $2 k^2$ times.

For one single unbiased birth-and-death chain it is well-known that (see
Lemma \ref{apx:bdlemma} in Appendix \ref{sec:apx:bdchain}), if the chain starts
at $0$ then it will hit the boundary after $k^2$ steps, in expectation. From
Markov inequality it thus follows that the probability that a single chain does
not hit the boundary within $2 k^2$ steps is at most $1/2$.

Let $\calE$ be the event $\calE =$ ``At round $t = 4 m k^2$ rounds all $e \in
[m]$ have been chosen at least $2 k^2$ times''. From~\eqref{eq:badevent} we
have that $\Prob{}{\calE} = 1 - m^{-\alpha + 1}$. Moreover, observe that
conditionally on event $\calE$, events $\{\tau(e) > 4mk^2 \,:\, e \in [m] \}$
are independent and each one of them has probability at most $1/2$ (since it is
the probability that a chain $e$ that has been chosen at least $2k^2$ times has
not reached the boundary yet). Hence,
\begin{align*}
& \Prob{}{\tau > 4mk^2} = \Prob{}{\tau > 4mk^2 \;|\; \calE} \Prob{}{\calE} + 
\Prob{}{\tau > 4mk^2 \;|\; \overline{\calE}} \Prob{}{\overline{\calE}} \\ 
& \leqslant \Prob{}{\tau > 4mk^2 \;|\; \calE} + \Prob{}{\overline{\calE}} \leqslant \Prob{}{\text{for all } e \in [m],\; \tau(e) > 4mk^2 \;|\; \calE} + m^{-\alpha + 1} \\
& \leqslant 2^{-m} + m^{-\alpha+1}
\end{align*}

\medskip As for the lower bound, from Chernoff bound (see
Lemma~\ref{apx:chernoff} in the Appendix) it follows that, for every $e = 1,
\dots, m$
\[
\Prob{}{\overline{Y}_t(e) \geqslant \frac{3t}{2m}} \leqslant e^{-\frac{9t}{4m}}\,.
\]
Hence, if we chose $t = m k^2 / (27 \log n)$, we have that for any $e = 1,
\dots, m$, the probability that $e$ has been chosen at least $k^2 / (18 \log
n$) times is at most
\[
\Prob{}{\overline{Y}_t(e) \geqslant k^2 / (18 \log n)} \leqslant e^{-k^2 / 12}\,.
\]
and by using the union bound, the probability that an $e$ exists that has been
chosen at least $k^2 / (18 \log n)$ times is
\begin{equation}\label{eq:badeventlb}
\Prob{}{\exists e \,:\, \overline{Y}_t(e)
\geqslant k^2 / (18 \log n)} \leqslant m e^{-k^2 / 12}\,.
\end{equation}
Thus, if $k > \sqrt{\alpha \log n}$ for some large enough constant $\alpha$, we
have that with probability at least $1-n^{-1}$ at round $t = m k^2 / (27 \log n)$
all $e \in \{1, \dots, m\}$ have been chosen at most $k^2 / (18 \log n)$ times.

For one single unbiased birth-and-death chain, if the chain starts at $0$ then
the probability that it hits the boundary at $\pm k$ within $k^2 / (18 \log n)$
steps is smaller than $1/n^3$ (see Lemma~\ref{lemma:lbhitlemma} in
Appendix~\ref{sec:apx:bdchain}). Hence, if we define event $\calE =$ ``At round
$t = m k^2 / (27 \log n)$ all $e \in [m]$ have been chosen at most $k^2 / (18
\log n)$ times'', then conditional on event $\calE$ we have that 
\begin{align*}
& \Prob{}{\tau \leqslant \frac{m k^2}{27 \log n}}
\leqslant \Prob{}{\left. \tau \leqslant \frac{m k^2}{27 \log n}\; \right| \; \calE} 
+ \Prob{}{\overline{\calE}}\\
& = \Prob{}{\left. \exists e \in \{ 1, \dots, m\} \,:\,\tau(e) \leqslant \frac{m k^2}{27 \log n} \right| \; \calE}
+ \Prob{}{\overline{\calE}}\\
& \leqslant m \Prob{}{\left. \tau(e) \leqslant \frac{m k^2}{27 \log n} \right| \; \calE} + 1 / n \\
& \leqslant m / n^3 + 1/n \leqslant 2/n \,.
\end{align*}

\end{proof}

%% file: trunk/betweenness.tex
\section{Edge-betweenness and failure time}\label{sec:betweenness}
In Section~\ref{sec:clique} we have seen that, when the underlying graph is
complete, the process described in Algorithm~\ref{alg:random_process} can be
coupled on the same probability space with a process involving the hitting time
of the boundary for a family of independent birth-and-death chains
(Algorithm~\ref{alg:multiple_bdc_process}). That is possible since any
shortest-path in a clique consists in a single edge, hence at each round the
balance of one single edge is updated and the updates of the balances of
different edges in different rounds are independent. In order to generalize the
results obtained for complete graph and constant capacities to the case of
different graph topologies and to the case in which edges can have different
initial capacities, the main technical difficulty stands in the fact that the
updates of the balances of edges that belong to the path chosen in a specific
round (line~4 in Algorithm~\ref{alg:random_process}) are not independent.
However, we can still give upper and lower bounds on the time it takes to have
the first payment failure, by losing only an extra $\bigO(\log n)$ factor
in the upper bound due to the dependence of the edges, and by using the
\textit{edge-betweenness centrality} as a parameter, in
Theorem~\ref{thm:beetweenness} we indeed prove that the critical quantity in
the estimation of the first failure time is the minimum of the ratios $k_e^2 /
g(e)$, where $k_e$ is the capacity of edge $e$ and $g(e)$ is its betweenness
centrality in graph $G$.

Edge betweenness is a commonly used centrality measure in network
analysis~\cite{girvan2002community}. It quantifies how often an edge lies in
shortest paths between pairs of nodes in a graph. Formally, the edge
betweenness centrality $g(e)$ of an edge $e$ is defined as
\begin{align}
g(e) =\sum_{s,t \in V} \frac{\sigma(s, t \, |\, e)}{\sigma(s, t)}
\end{align}
where $V$ is the set of nodes, $\sigma(s, t)$ is the number of shortest $(s,
t)$-paths, and $\sigma(s, t \,|\, e)$ is the number of those paths containing
edge $e$.

In our random process (Algorithm~\ref{alg:random_process}), the probability
that an edge $e$ is contained in the shortest path selected at any specific
iteration is, by the law of total probability,
\begin{align}
\Prob{}{e} = \sum_{s,t \in V} \Prob{}{e \,|\, s, t} \cdot \Prob{}{s, t}
\end{align}
where $\Prob{}{e \,|\, s, t}$ is the probability of picking a shortest path
containing $e$ at line $4$ in Algorithm~\ref{alg:random_process} conditional on
the fact that $s$ and $t$ were selected as source and destination at lines $2$
and $3$, and $\Prob{}{s, t}$ is the probability of selecting $s$ and $t$ as
source and destination of the path.  Notice that the probability of picking $e$
given $s$ and $t$ is exactly $\frac{\sigma(s, t \,|\, e)}{\sigma(s, t)}$ and
$\Prob{}{s, t} = \frac{2}{n(n-1)}$, since we select the pair of nodes uniformly
at random. Thus, $\Prob{}{e}$ is the normalized edge-betweenness of edge $e$
\begin{equation}\label{eq:probasbetwenness}
\Prob{}{e} = \frac{2}{n(n-1)} g(e)
\end{equation}

\begin{theorem}\label{thm:beetweenness}
Let $G = (V,E)$ be a graph with $n$ nodes. For every channel $e \in E$, let
$c(e) = 2k_e$ be its capacity, with $k_e \in \mathbb{N}$ and $k_e > \alpha
\sqrt{\log n}$ for a sufficiently large constant $\alpha$, and let
$b_e^{(0)}(u) = b_e^{(0)}(v) = k_e$ be its initial balance. Let $\tau$ be the
random variable indicating the first time one of the channels is empty in one
direction, i.e., $\tau = \inf\{t \in \mathbb{N} \,:\, \exists e \in E \text{
with } b_{\min}(e) = 0 \}$. Then, 
\begin{enumerate}
\item $\tau = \bigO\left(\xi \, n^2 \log n\right)$, w.h.p.
\item $\tau = \Omega\left(\xi \, \frac{n^2}{\log n}\right)$, w.h.p.
\end{enumerate}
where $\xi = \min \left\{ \frac{k_e^2}{g(e)} \,:\, e \in E\right\}$ and $g(e)$
is the betweenness centrality of edge $e$.
\end{theorem}

\begin{proof}
Let $e$ be an edge, let $Y_t(e)$ be the indicator random variable of the event
``Edge $e$ is included in the path chosen at round $t$'', and let 
\[
\overline{Y}_t(e) = \sum_{i=1}^t Y_i(e)
\]
From~\eqref{eq:probasbetwenness} we have that $\Prob{}{Y_i(e) = 1} =
\frac{2}{n(n-1)} g(e)$ for every $i$ and thus the expected number of rounds
edge $e$ has been included in payment paths up to round $t$ is
\begin{equation}\label{eq:edgepathincl}
\Expec{}{\overline{Y}_t(e)} = t \cdot \frac{2}{n(n-1)} g(e)\,. 
\end{equation}
Let $\tau_e$ be the random variable indicating the first time edge $e$ is empty
in one direction
\[
\tau_e = \inf\left\{t \in \mathbb{N} \,:\, \bmin{e} = 0 \right\}
\]
Notice that, restricted only to the rounds in which edge $e$ is included in a
path, the balance of edge $e$ behaves according to an unbiased birth-and-death
chain. 

\smallskip As for the upper bound, for an edge $e$ with capacity $k_e$, the
expected time to become empty in one direction is at most $k_e^2$ rounds in
which edge $e$ is included in the path chosen at that round, regardless of the
initial balance of the edge (see Lemma~\ref{apx:bdlemma} in
Appendix~\ref{sec:apx:bdchain}). From Markov inequality it thus follows that,
if at round $t$ an edge $e$ updated its balance more than $4k_e^2$ rounds, then
the probability that it is not yet empty in one direction is at most $1/4$,
\[
\Prob{}{\tau_e > t \;|\; \overline{Y}_t(e) > 4 k_e^2 } \leqslant \frac{1}{4}
\]
Hence, if we take 
\[
\bar{t}_e = 4 n(n-1) \frac{k_e^2}{g(e)}
\]
from~\eqref{eq:edgepathincl} we have that
$\Expec{}{\overline{Y}_{\bar{t}_e}(e)} = 8 k_e^2$ and, since $\{Y_i(e) \,:\, i
= 1, \dots, \bar{t}_e \}$ are independent, from Chernoff bound it follows that
\[
\Prob{}{\overline{Y}_{\bar{t}_e}(e) \leqslant 4 k_e^2} 
\leqslant e^{-(1/2) n (n-1) \left(k_e^2 / g(e) \right)}
\leqslant \frac{1}{4}
\]
Hence
\begin{align*}
  \lefteqn{ \Prob{}{\tau_e > \bar{t}_e} }  \\
  &= \Prob{}{\tau_e > \bar{t}_e \mid \overline{Y}_{\bar{t}_e}(e) > 4 k_e^2 } \,
     \Prob{}{\overline{Y}_{\bar{t}_e}(e) > 4 k_e^2 } 
     + \Prob{}{\tau_e > \bar{t}_e \mid \overline{Y}_{\bar{t}_e}(e) \leqslant 4 k_e^2 } \,
     \Prob{}{\overline{Y}_{\bar{t}_e}(e) \leqslant 4 k_e^2 } \\
  &\leqslant \Prob{}{\tau_e > \bar{t}_e \mid \overline{Y}_{\bar{t}_e}(e) > 4 k_e^2 }
     + \Prob{}{\overline{Y}_{\bar{t}_e}(e) \leqslant 4 k_e^2 } 
   \leqslant \frac{1}{2} 
\end{align*}
Since the above bound holds regardless of the initial balance of edge $e$, 
it follows that
\[
\Prob{}{\tau_e > \bar{t}_e \log n} 
\leqslant \left(\frac{1}{2}\right)^{\log n} = \frac{1}{n}
\]
Hence, for every edge $e \in E$, w.h.p. it holds that
\[
\tau_e \leqslant 4 n(n-1) \frac{k_e^2}{g(e)} \log n 
= \bigO\left(n^2 \log n  \frac{k_e^2}{g(e)} \right)
\]
The upper bound follows considering the edge(s) $e$ with the smallest ratio
$k_e^2 / g(e)$.

\smallskip As for the lower bound, we can proceed as in the proof of
Theorem~\ref{thm:clique}. For an edge $e$ with capacity $2k_e$, if $e$ starts
from a perfectly balanced configuration, from Lemma~\ref{lemma:lbhitlemma} in
Appendix~\ref{sec:apx:bdchain} it follows that, at any time $t$,
\begin{equation}\label{eq:condtauelb}
\Prob{}{\tau_e \leqslant t \;\left|\; \overline{Y_t}(e) 
\leqslant \frac{k_e^2}{18 \log n} \right.} 
\leqslant 4 e^{\frac{-k_e^2}{k_e^2 / (3 \log n)}} = \frac{4}{n^3}
\end{equation}
If we take
\[
\bar{t}_e = \frac{n(n-1)}{54 \log n} \cdot \frac{k_e^2}{g(e)}
\]
from~\eqref{eq:edgepathincl} we have that 
\[
\Expec{}{\overline{Y}_{\bar{t}_e}(e)} = \frac{k_e^2}{27 \log n}
\]
and from Chernoff bound (Lemma~\ref{apx:chernoff} in
Appendix~\ref{sec:apx:concentration} with $\delta = 1/2$) 
\begin{equation}\label{eq:badeventlbbetwenn}
\Prob{}{\overline{Y}_{\bar{t}_e}(e) \geqslant \frac{k_e^2}{18 \log n}} 
\leqslant e^{-\frac{1}{12} \frac{k_e^2}{27 \log n}} 
\leqslant \frac{1}{n^2}
\end{equation}
where in the last inequality we used the hypothesis $k_e > \alpha \sqrt{\log n}$
for a sufficiently large constant $\alpha$.

Now let $\bar{t}$ be 
\begin{equation}\label{eq:targetlbbetweenness}
\bar{t} 
= \frac{n(n-1)}{54 \log n} \cdot \min\left\{\frac{k_e^2}{g(e)} \,:\, e \in E \right\}
\end{equation}
and let $\calE$ be the event $\calE =$ ``At round $\bar{t}$ for each edge $e
\in E$ it holds that $\overline{Y}_{\bar{t}}(e) < \frac{k_e^2}{18 \log n}$''.
Since $\bar{t} \leqslant \bar{t}_e$ for every $e \in E$,
from~\eqref{eq:badeventlbbetwenn} it follows that
$\Prob{}{\overline{Y}_{\bar{t}}(e) \geqslant \frac{k_e^2}{18 \log n}} \leqslant
1/n^2$. Hence,
\begin{equation}\label{eq:bet:globalbadeventlb}
\Prob{}{\overline{\calE}} 
= \Prob{}{\exists e \in E \;:\; \overline{Y}_{\bar{t}}(e) \geqslant \frac{k_e^2}{18 \log n}}
\leqslant \frac{1}{n}
\end{equation}
The first time an edge is empty in one direction is $\tau = \min\{\tau_e \,:\;
e \in E\}$. For $\bar{t}$ defined as in \eqref{eq:targetlbbetweenness} we thus
have
\begin{align*}
\Prob{}{\tau < \bar{t}}
& \leqslant \Prob{}{\tau < \bar{t} \;\left|\; \calE \right.} 
+ \Prob{}{\overline{\calE}} \leqslant \Prob{}{\exists e \in E \,:\, \tau_e < \bar{t} \;\left|\; \calE \right.}
+ \Prob{}{\overline{\calE}} \\ 
& \leqslant \Prob{}{\exists e \in E \,:\, \tau_e < \bar{t}_e \;\left|\; \calE \right.}
+ \Prob{}{\overline{\calE}} \leqslant \frac{4}{n} + \frac{1}{n}
\end{align*}
where in the last inequality we used the bounds in~\eqref{eq:condtauelb} and
in~\eqref{eq:bet:globalbadeventlb}.
\end{proof}

Theorem~\ref{thm:beetweenness} generalizes the result on the complete graph of
Theorem~\ref{thm:clique} with respect to both graph topology and channel
capacity. Notice that, when all edge capacities are equal, $k_e = k$ for all $e
\in E$, then upper and lower bounds in Theorem~\ref{thm:beetweenness} turn out
to be, w.h.p., 
\[
\tau = \left\{
\begin{array}{l}
\Omega\left(\frac{k^2}{g_{\max}} n^2 \log n \right) \\[2mm]
\mathcal{O}\left( \frac{k^2}{g_{\max}} \cdot \frac{n^2}{\log n} \right)
\end{array}
\right.
\]
where with $g_{\max} = \max\{g(e) \,:\, e \in E \}$ is the largest edge
betweenness centrality. The above observation indicates that, despite the 
fact that the updates of the edges in general graphs are not independent,
the impact of such correlation can only be limited to a $\log n$ factor,
as a function of $n$. For fixed $n$ and as a function of the edge capacity 
$k$ the impact of the correlation is limited to a constant factor. In the
next section we will show the results obtained with some simulations that
highlight such constant-factor impact. 

%% file: trunk/simulations.tex
\section{Simulations}\label{sec:simulations}
In this section, we present the results of the simulations of the random
process in Algorithm~\ref{alg:random_process} on different classes of graphs.

We begin with the case of the complete graph (see
Subsection~\ref{ssec:sim:complete}), where the goal is to compare the results
of the simulations with the theoretical upper and lower bounds in
Section~\ref{sec:clique}. As we will see in Figure~\ref{fig:Clique}, the
results of the simulations give some evidence that the lower bound in
Theorem~\ref{thm:clique} could be tight.

To highlight the impact of the correlation between the updates of the balances
of different edges, we then proceed with a comparison of the results of the
simulations obtained in a graph with $n$ nodes and a ring topology with the
results that we get on an unstructured set of $n$ edges (see
Subsection~\ref{ssec:sim:ring}): In a graph with a ring topology, every time we
pick a source and a destination uniformly at random and we pick the shortest
path between them to route the payment, for any fixed edge $e$ the probability
that $e$ is included in the path is about $1/4$, but edges are correlated
(two edges that are close in the ring topology have larger probability to be
updated together than two edges that are far apart); for comparison, we use an
unstructured set of $n$ edges, where at each round we pick each edge with the
same probability, about $1/4$, independently of the other edges, and we
independently update the balances of all the picked edges. 

Finally, to highlight the impact of the distribution of the channel capacities
in the Lightning Network, we run the simulation of our random process on a
snapshot of the channel graph of the real Lightning Network (see
Subsection~\ref{ssec:sim:lnd}): We first consider each edge having its real
channel capacity, as observed in the Lightning Network; then, we run the
simulations on a graph whose topological structure is the same but the
capacities of the channels are suitably redistributed, in order to quantify the
improvement that we would get, with respect to the time to have the first
payment failure in our model, if it was possible to modify the channel
capacities.

The code for all simulations have been written in C++. The simulations were
conducted on a system equipped with an AMD Ryzen 5 CPU (3.9 GHz) and $32$ GB of
RAM. The simulations in Subsection~\ref{ssec:sim:lnd} refer to a snapshot of
the Lightning Network collected on 2025/04/23 using our Lightning Network node,
that runs an instance of LND~\cite{lnd_github}.

\subsection{Complete graph}\label{ssec:sim:complete}
In this subsection we present the results of the simulations of the random
process on the complete graph. In Figure~\ref{fig:Clique} we show the results
that we get when we fix the capacity of each edge at $512$ and we increase the
number of nodes $n$ up to $2800$; for each $n$, we run the simulation $10$
times. The plots in the figure show average, maximum and minimum failure time,
for each value of $n$, over the $10$ runs. To compare the results of the
simulations with the theoretical results in Theorem~\ref{thm:clique} we also
plot the lines that we get for the theoretical upper and lower bounds, where
the constants hidden in the asymptotic notation is chosen using nonlinear least
squares fitting~\cite{Lain1978Lsquares}, i.e., minimizing the squared error
with the simulated average. Namely, for the upper bound we plot the function
$\tau = p \cdot (512)^2 n^2$, with $p = 0.0044$, and for the lower bound we
plot $\tau = p \cdot \frac{(512)^2 n^2}{\log(N)}$, with $p = 29.0202$.  The
growth rate of the average of the simulations seems much closer to the growth
rate of the theoretical lower bound than that of the upper bound.

\begin{figure}[ht]
    \centering
    \includegraphics[width=\linewidth]{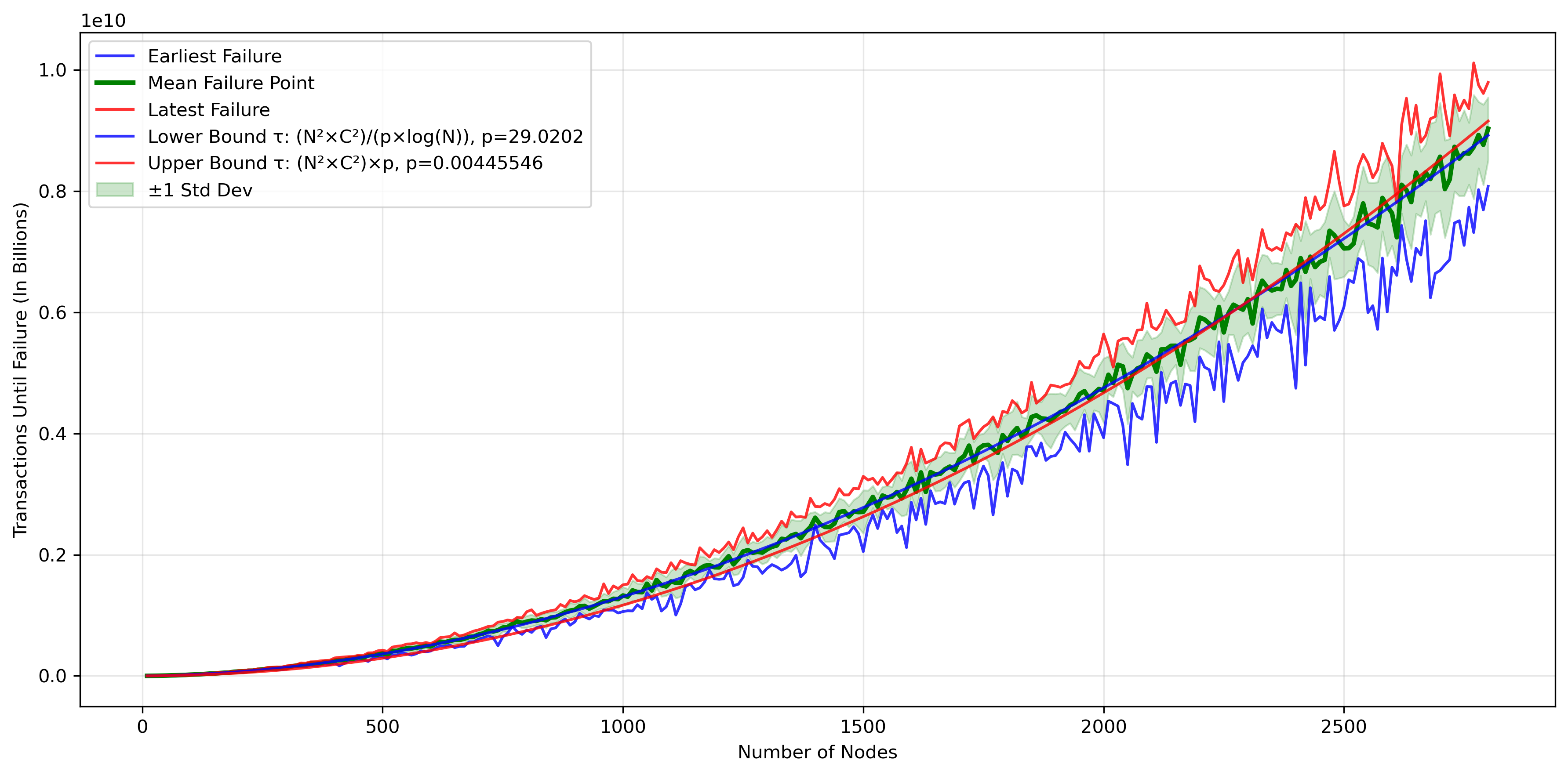}
    \caption{The payment throughput in a clique with fixed channel capacity
        of 512. Lines depict earliest (blue), average (green), and latest (red)
        failure points, with theoretical bounds (dashed).The green band
        indicates standard deviation, demonstrating how payment capacity
        scales non-linearly with network size.}
    \label{fig:Clique}
\end{figure}

\subsection{Ring topology \& Independent edges}\label{ssec:sim:ring}
In this subsection we consider a ring with $n$ nodes and a system of $n$
independent chains. On a graph with a ring topology, when we sample source and
destination of the payment u.a.r. among all nodes, it is easy to see that for
each $i = 1, 2, \dots, \lceil \frac{n}{2} \rceil - 1$, the probability that
source $u$ and destination $v$ are at distance $i$ is $\Prob{}{d(u, v) = i} =
{n}/\binom{n}{2} = \frac{2}{n-1}$. For a given distance $i$ and a given edge
$e$, the probability that $e$ is included in the chosen path conditional on the
fact that source and destination are at distance $i$ is $\Prob{}{e \;|\; d(u,
v) = i} = \frac{i}{n}$.

\begin{equation}\label{eq:probedgering}
    \Prob{}{e} = \sum_{i=1}^{\lceil \frac{n}{2} \rceil -1} \frac{2}{n-1}  \cdot \frac{i}{n}
    = \frac{(\lceil \frac{n}{2} \rceil -1)\lceil \frac{n}{2} \rceil}{n(n-1)} \simeq \frac{1}{4}
\end{equation}

On a ring topology (as well as in almost any other graph topologies) the update
of the balances of edges involved in the same path in a given round of the
random process are not independent, since the balances of all edges in the path
are updated in the same \textit{direction}: The balance of every edge $e$ in
the path decreases of one unit for the endpoint closer to the source and it
increases of one unit for the endpoint closer to the destination.  On the other
hand, on a system of $n$ independent chains, where at every round every chain
is selected with probability given by~\eqref{eq:probedgering} and each selected
chain updates its balance independently of the other chains, the expected
number of chains that update their balance equals the expected number of edges
that update their balance at each round in the ring, but the updates of the
balances in this system are completely independent.

Figure~\ref{fig:indp_cir} shows the results that we get when we fix the number
of nodes at $4096$, all edges have the same capacity, and the capacity of each
edge increases up to $3040$. For each value of the capacity we run $10$
simulations and we plot maximum, minimum and average first failure time over
the runs, for the case in which edges are disposed according to a ring topology
and for the case in which edges are independent. The plots show that in the
case of the ring topology the growth rate of the average is faster than the
growth rate of the independent case. It is also interesting to note that the
results for the ring topology exhibit a larger variance than the results for
the independent chains.

\begin{figure}[ht]
    \centering
    \includegraphics[width=\linewidth]{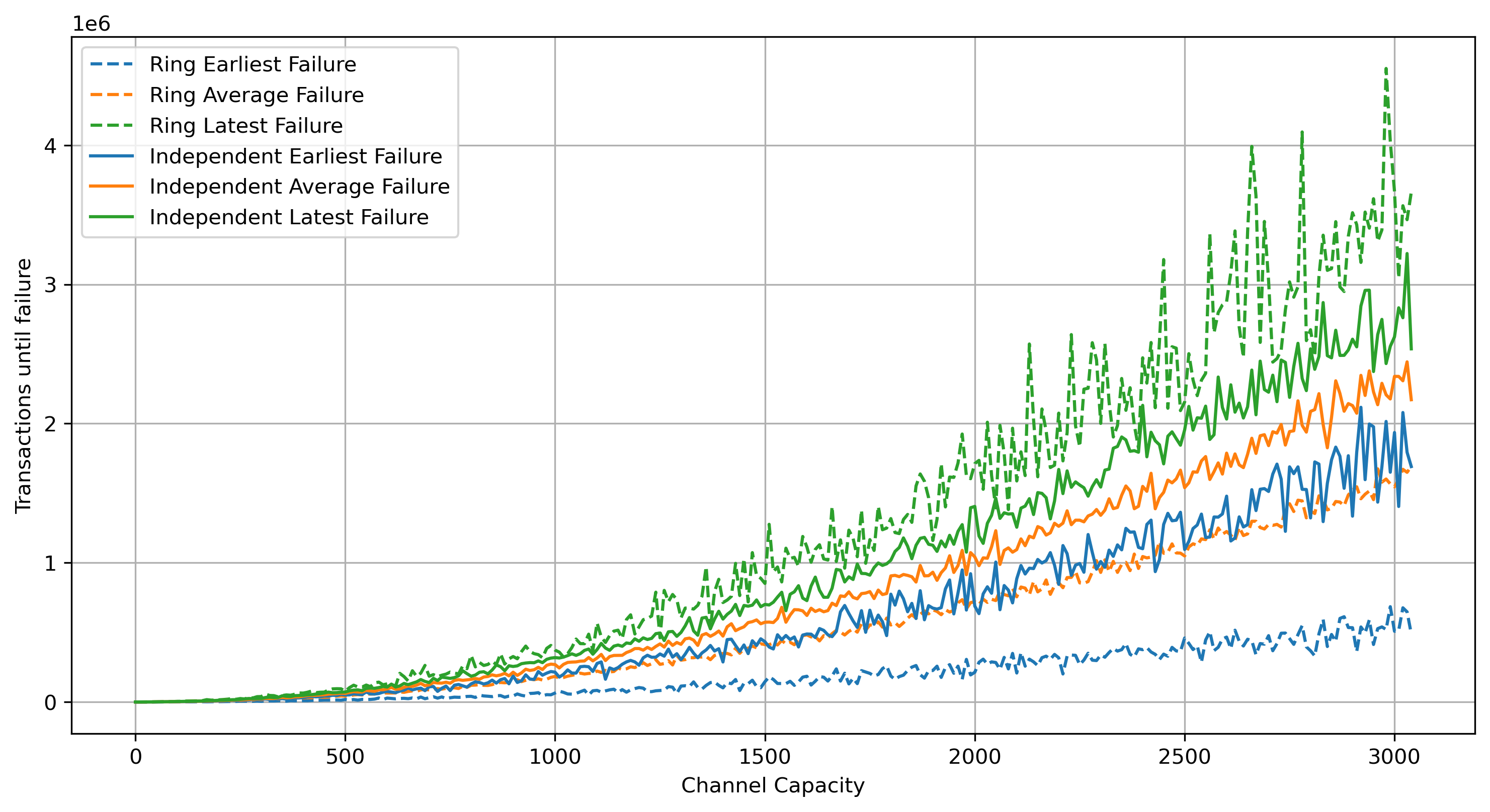}
    \caption{Simulated failure probabilities for ring and independent chain
        systems with $n=4096$ nodes and maximum capacity $3040$.}
    \label{fig:indp_cir}
\end{figure}

\subsection{Lightning Network}\label{ssec:sim:lnd}
In this subsection we present the results of the simulations of our random
process on a snapshot of the channel graph of the real Lightning Network, the
snapshot was taken on April 23rd, 2025 and it consists of $16426$ nodes and
$51063$ edges. Since the whole graph is not connected, in the simulation we
only consider the giant component, that consists in $14900$ nodes and $47087$
edges. Furthermore, as the average channel capacity of the network is about
$10$ million \textit{satoshi}\footnote{In the Lightning Network, channel
    capacities and payment values are usually expressed in \textit{satoshi} and
    \textit{millisatoshi}; recall that $1$ \textit{bitcoin} corresponds to
    $100$million \textit{satoshi}}, we use payment units of $1k$, $10k$, and $100k$
\textit{satoshi}. For each unit, we measure the number of payments between
randomly chosen source-destination pairs before the first payment failure occurs.
Table~\ref{tab:normal_Lightning} summarizes the results of the simulations,
recording maximum, minimum, average, and standard deviation over $10$ runs.

\begin{table}
    \centering
    \begin{tabular}{ccccc}
        \hline
        TX amount & min & max & mean & std  \\
        \hline
        1k        & 1   & 108 & 13   & 12   \\
        10k       & 1   & 30  & 4    & 3.63 \\
        100k      & 1   & 12  & 1    & 0.95 \\
        \hline
\end{tabular}
\caption{Time failure counts for the unmodified graph of the Lightning Network
with different payment units: $1k$, $10k$ and $100k$ satoshis, where each
capacity simulation was repeated $1000$ times.}
\label{tab:normal_Lightning}
\end{table}

In a second set of simulations, we consider the same Lightning Network graph
structure, but we redistribute the overall network capacity evenly over the
edges, so that every edge has the same capacity and the sum of the capacities
is preserved. Table~\ref{tab:uniform_Lightning} summarizes the results of this
set of simulations.
\begin{table}
    \centering
    \begin{tabular}{ccccc}
        \hline
        TX amount & min      & max       & mean      & std       \\
        \hline
        1k        & 74833576 & 968962509 & 439701719 & 204199218 \\
        10k       & 765772   & 9801570   & 3862534   & 1617777   \\
        100k      & 8084     & 115597    & 40089     & 17233     \\
        \hline
\end{tabular}
\caption{Time failure counts for the unmodified graph with uniform
redistribution capacity of edges with payment units: $1k$, $10k$, and $100k$
satoshis, where each capacity simulation was repeated $1000$ times.}
\label{tab:uniform_Lightning}
\end{table}

Finally, in a third set of simulations, we consider the same Lightning Network
graph structure, but we redistribute the overall network capacity in a way that
is \textit{optimized} with respect to our random process where we take source
and destination nodes uniformly at random over all the nodes: More precisely,
the overall network capacity is redistributed over the edges so that the ratio
$k_e^2 / g(e)$, the square of the capacity and the betweenness centrality, is
the same for every edge.
Table~\ref{tab:optimization} presents the results for the simulations with
payments of $100k$ satoshis. Running the simulations with payments of $10k$ and
$1k$ satoshis would take too long with such redistribution of the channel
capacity. A comparison of the results for the case of payments of $100k$
satoshis in the three sets of simulations is summarized in
Figure~\ref{fig:Distribution_of_transaction}: We can observe that the network
essentially cannot handle transactions of $100k$ satoshis with its original
capacity distribution; if we redistribute the overall capacity evenly over all
edges, then the network becomes more resistant with respect to failure;
finally, if we redistribute the overall capacity so the ratio $k_e^2 / g(e)$ is
equal for all edges, then there is a difference of almost $20$ times on
average, based on the results we showed in Table~\ref{tab:optimization} and
Table~\ref{tab:uniform_Lightning}. Clearly, since it is not possible to
organize and modify the capacity of the channels of the Lightning Network in a
centralized way, all the above redistributions are to be intended as thought
experiments.

\begin{table}
    \centering
    \begin{tabular}{cccccc}
        \hline
        TX amount & min & max     & mean   & std    \\
        \hline
        100k      & 841 & 2144747 & 782110 & 521689 \\
        \hline
    \end{tabular}
    \caption{Time failure counts for the unmodified graph with optimization on redistribution capacity of edges with payment unit: $100k$ satoshis, where the simulation was repeated 1000 times.}
    \label{tab:optimization}
\end{table}

\begin{figure}[ht]
    \centering
    \includegraphics[width=\linewidth]{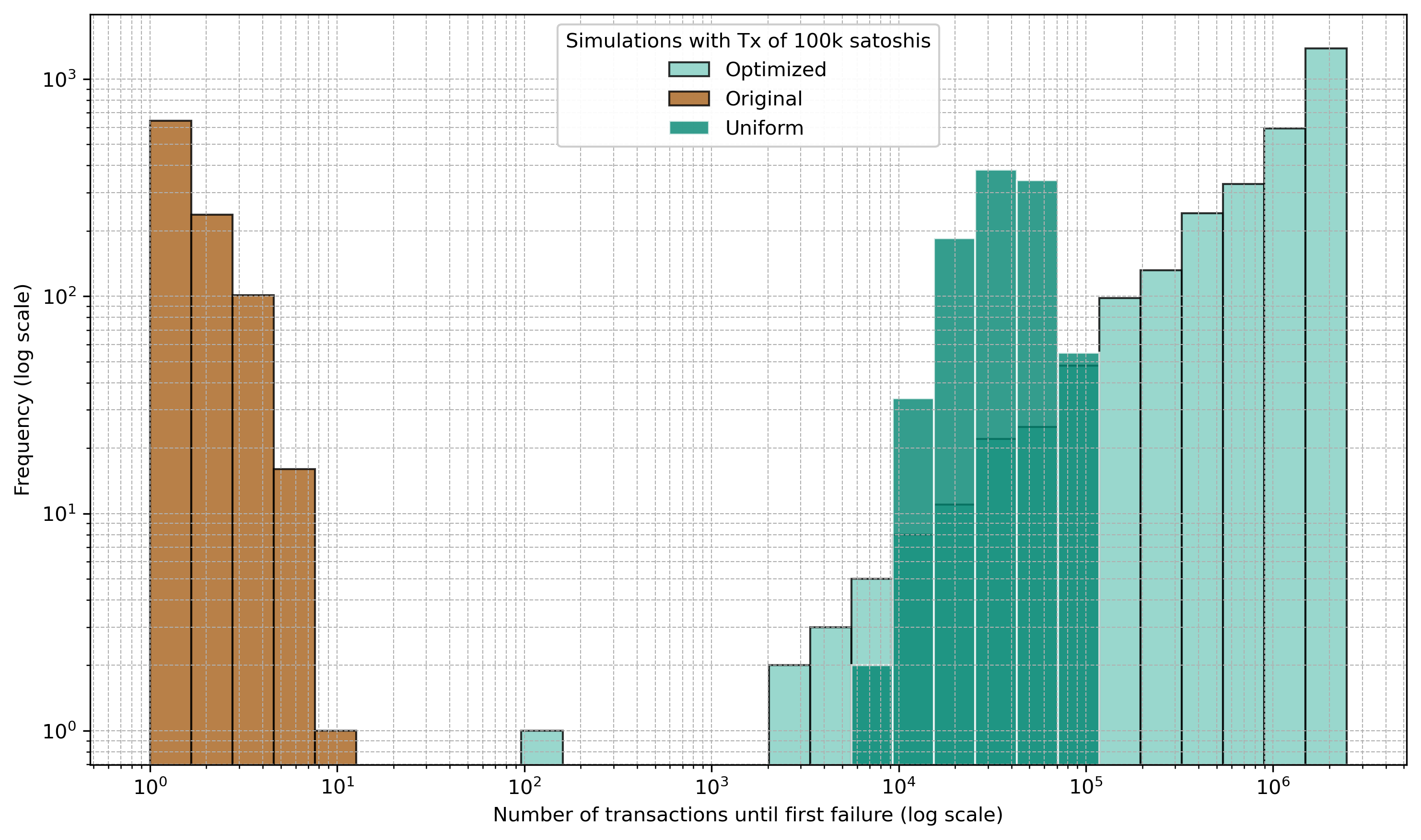}
    \caption{Log-log histogram of simulations for sending payments of value $100k$ satoshis with different distribution each simulation was repeated 1000 times}
    \label{fig:Distribution_of_transaction}
\end{figure}

%% file: trunk/concl.tex
\section{Conclusions}\label{sec:conclusions}

In this paper we studied a random process defined over undirected graphs
inspired by the problem of payment failures on the channel graph of the
Lightning Network.

We first analyzed upper and lower bounds on the time it takes to have the first
payment failure when the underlying graph is complete and all edges have the
same capacity, as a function of the number of nodes in the graph and of the
capacity of the edges. Our upper and lower bounds are almost tight, except for
a gap of $\log n$, that would be interesting to close. Our conjecture is that
the upper bound can be improved but the lower bound is tight.

In our random process, we pick source and destination of a payment uniformly at
random among all nodes in the graph, and we pick a path between source and
destination uniformly at random among all shortest paths between them. It
implies that, in general graphs, the larger the betweenness centrality of an
edge, the more often the edge appears in the selected paths. Since the number
of times an edge has to be selected before a payment failure occurs is
proportional to the square of the capacity, the critical quantity for the
random process is thus the ratio $k_e^2 / g(e)$ of the square of the capacity
of an edge and its betweenness centrality. For arbitrary graphs we thus give
upper and lower bounds on the time it takes to have a payment failure as
functions of the number of nodes and of the minimum over all the edges of such
ratios.

Our simulations on the complete graph validate the theoretical results and give
some empirical evidence that the theoretical lower bound might be tight. The
simulations on the ring give hints about the impact that the dependency among
the balance updates of the edges belonging to the same path has on the time it
takes for the first payment failure to occur. Finally, our simulations on the
channel graph of the Lightning Network quantify the impact of the distribution
of the channel capacity over the edges on the transaction failure rate.

%% file: trunk/apx.tex

\section{Concentration inequalities}\label{sec:apx:concentration}
In this section, we state a fundamental concentration bound that we used in the
proofs in Sections~\ref{sec:clique} and~\ref{sec:betweenness}. We refer the
reader interested in a simple proof to~\cite{dubhashi2009concentration}
or~\cite{mitzenmacher2017probability}.

\begin{lemma}[Chernoff bound~\cite{chernoff1952measure}]\label{apx:chernoff}
    Let $\{X_i \;:\, i = 1, \dots, n\}$ be a family of independent Bernoulli
    random variables with $\Prob{}{X_i = 1} = p_i$ and let $X = \sum_{i=1}^n X_i$.
    Then, for every $0 < \delta < 1$ it holds that
    \begin{itemize}
        \item $\Prob{}{X \leqslant (1-\delta)\mu} \leqslant e^{-\frac{\delta^2}{2} \mu}$
        \item $\Prob{}{X \geqslant (1+\delta)\mu} \leqslant e^{-\frac{\delta^2}{3} \mu}$
    \end{itemize}
    where $\mu = \Expec{}{X} = \sum_{i = 1}^n p_i$.
\end{lemma}

\section{Preliminaries on birth-and-death chains}\label{sec:apx:bdchain}
In this section we briefly recall some terminology and fundamental bounds about
\textit{birth-and-death} Markov chains that we used in
Sections~\ref{sec:clique} and~\ref{sec:betweenness}.

A finite birth-and-death chain with absorbing boundary is a Markov chain
$\{X_t\,:\, t \in \mathbb{N}\}$ with state space
\[
\Omega = \{-k, -k+1, \dots, -1, 0, 1, \dots, k-1, k\}
\]
for some $k \in \mathbb{N}$, where 
\[
\Prob{}{X_{t+1} = -k \,|\, X_t = -k} = \Prob{}{X_{t+1} = k \,|\, X_t = k} = 1 \\
\]
and for each $j = -k+1, \dots, -1, 0, 1, \dots, k-1$,
\[
\begin{array}{lcl}
  \Prob{}{X_{t+1} = j + 1 \mid X_t = j}  =  p \\
  \Prob{}{X_{t+1} = j - 1 \mid X_t = j}  = q \\
  \Prob{}{X_{t+1} = j \mid X_t = j} =  1 - (p + q)
\end{array}
\]
for some non-negative $p$ and $q$ with $p + q < 1$. When $p = q = 1/2$ the
chain is said to be \textit{unbiased}.

\begin{center}
    \begin{small}
        \begin{tikzpicture}[scale=0.7]
            \node[circle,draw,color=blue,label=above:\textcolor{blue}{\small{$0$}}] (sj) at (0,0) {};
            \node[circle,draw,color=blue,label=above:\textcolor{blue}{\small{$1$}}] (sj+1) at (2,0) {};
            \node[circle,draw,color=blue,label=above:\textcolor{blue}{\small{$-1$}}] (sj-1) at (-2,0) {};
            \node[color=blue] (sdot) at (4,0) {$\dots\dots$};
            \node[circle,draw,color=blue,label=below:\textcolor{blue}{\small{$k$}}] (sn) at (6,0) {};
            \node[color=blue] (s-dot) at (-4,0) {$\dots\dots$};
            \node[circle,draw,color=blue,label=below:\textcolor{blue}{\small{$-k$}}] (s0) at (-6,0) {};
            \draw[->,>=stealth] (sj) edge [bend left] node[above] {$1/2$} (sj+1)
            (sj-1) edge [bend left] (sj);
            \draw[->,>=stealth] (sj+1) edge [bend left] (sdot)
            (sj+1) edge [bend left] (sj);
            \draw[->,>=stealth] (sj-1) edge [bend left](s-dot)
            (sj) edge [bend left] node[below] {$1/2$} (sj-1);
            \draw[->,>=stealth] (sn) edge [loop above] node[above] {$1$} (sn)
            (s-dot) edge [bend left] (s0)
            (s0) edge [loop above] node[above] {$1$} (s-dot)
            (sdot) edge [bend left] (sn);
        \end{tikzpicture}
    \end{small}
\end{center}

For unbiased birth-and-death chains it is well known that the expected time
before the chain hits the boundary starting from an arbitrary state $j$ is $k^2
- j^2$ (see, e.g., Chapter~2 in~\cite{levine2009markov}).
\begin{lemma}\label{apx:bdlemma}
    Let $\tau = \inf \{t \in \mathbb{N} \,:\, X_t = \pm k \}$ be the random variable
    indicating the first time in which the chain hits state $k$. The
    expectation of $\tau$ for the chain starting at an arbitrary state $j \in \Omega$ is
    \[
        \Expec{}{\tau \,|\, X_0 = j} = k^2 - j^2
    \]
\end{lemma}

To bound the probability that a finite chain hits the boundary before or after
a certain value, it is sometimes convenient to consider infinite
birth-and-death chains. The following Lemma~\ref{lemma:hittpos} is a well-known
bound (see, e.g., Exercise~2.10 in~\cite{levine2009markov}) for infinite chains
that can be used to give an upper bound on the probability that a chain hits
the boundary within a certain number of time steps, as shown in
Lemma~\ref{lemma:lbhitlemma}.

For an event $\calE$ and a state $j \in \Omega$ we use the standard notation
$\Prob{j}{\calE}$ for the probability of $\calE$ conditional on the event that
the Markov chains $\{X_t\}$ starts at state $j$,
\[
\Prob{j}{\calE} = \Prob{}{\calE \;|\; X_0 = j} \,.
\]

\begin{lemma}\label{lemma:hittpos}
Let $\{X_t \,:\, t \in \mathbb{N} \}$ be an unbiased birth-and-death with state
space $\Omega = \mathbb{Z}$. For every time $t \in \mathbb{N}$ and for every
\textit{target state} $k \in \mathbb{N}$ it holds that
\[
\Prob{0}{|X_t| \geqslant k}
\leqslant \Prob{0}{\max_{i = 1, \dots, t} |X_i| \geqslant k}
\leqslant 2 \Prob{0}{|X_t| \geqslant k}
\]
\end{lemma}

The following lemma is used in the proofs of the lower bounds in
Theorems~\ref{thm:clique} and~\ref{thm:beetweenness}.

\begin{lemma}\label{lemma:lbhitlemma}
Let $\{X_t \,:\, t \in \mathbb{N} \}$ be an unbiased birth-and-death chain with
state space $\Omega = \mathbb{Z}$, let $k \in \mathbb{N}$ be an integer and let
$\tau$ be the random variable indicating the first time the chain hits state
$\pm k$,
\[
\tau = \inf\{t \in \mathbb{N} \,:\, |X_t| = k \}\,.
\]
Then, for every $t \in \mathbb{N}$ it holds that
\[
\Prob{0}{\tau \leqslant t} \leqslant 4 e^{-k^2 / (6t)}
\]
\end{lemma}
\begin{proof}
Observe that, the chain hits states $\pm k$ within time $t$ if and only if the
maximum over all times $i$ up to $t$ of $|X_{i}|$ is at least $k$,
\begin{align*}
\Prob{0}{\tau \leqslant t}
& = \Prob{0}{\max_{i = 1, \dots, t} |X_i| \geqslant k}
\leqslant 2 \Prob{0}{|X_t| \geqslant k}                 
= 4 \Prob{0}{X_t \geqslant k}
\end{align*}
where in the second inequality we used Lemma~\ref{lemma:hittpos} and in the
third equality we used the symmetry of the process, i.e., $\Prob{0}{X_t
\geqslant +k} = \Prob{0}{X_t \leqslant -k}$.

Now observe that, since the state of the chain at each round increases of one
unit with probability $1/2$ and decreases of one unit with probabiity $1/2$,
then the probability that the state of the chain at time $t$ is at least $k$
equals the probability that in a sequence of $t$ unbiased coin tosses at least
$(t + k) / 2$ ended up head,
\[
\Prob{0}{X_t \geqslant k} = \Prob{}{B(t,1/2) \geqslant \frac{k+t}{2}}
\]
where with $B(t,1/2)$ we indicated a binomial random variable with parameters
$t$ and $1/2$. The thesis then follows from the Chernoff bound (see
Lemma~\ref{apx:chernoff}).
\end{proof}



%% file: mls.bib
@book{antonopoulos2022mastering,
  title     = {Mastering the Lightning Network: A Second Layer Blockchain Protocol for Instant Bitcoin Payments},
  author    = {Antonopoulos, A.M. and Osuntokun, O. and Pickhardt, R.},
  url       = {https://github.com/lnbook/lnbook},
  year      = {2021},
  publisher = {O'Reilly}
}

@article{chernoff1952measure,
  title={A measure of asymptotic efficiency for tests of a hypothesis based on the sum of observations},
  author={Chernoff, Herman},
  journal={The Annals of Mathematical Statistics},
  pages={493--507},
  year={1952},
  publisher={JSTOR}
}

@inproceedings{croman2016scaling,
  title={On Scaling Decentralized Blockchains: (A Position Paper)},
  author={Croman, Kyle and Decker, Christian and Eyal, Ittay and Gencer, Adem Efe and Juels, Ari and Kosba, Ahmed and Miller, Andrew and Saxena, Prateek and Shi, Elaine and G{\"u}n Sirer, Emin and others},
  booktitle={International conference on financial cryptography and data security},
  pages={106--125},
  year={2016},
  organization={Springer}
}

@book{dubhashi2009concentration,
  title={Concentration of measure for the analysis of randomized algorithms},
  author={Dubhashi, Devdatt P. and Panconesi, Alessandro},
  year={2009},
  publisher={Cambridge University Press}
}

@book{levine2009markov,
  title={Markov Chains and Mixing Times},
  author={Levin, David A. and Peres, Yuval and Wilmer, Elizabeth L.},
  publisher={American Mathematical Society, Providence},
  year={2009}
}

@book{mitzenmacher2017probability,
  title={Probability and computing: Randomization and probabilistic techniques in algorithms and data analysis},
  author={Mitzenmacher, Michael and Upfal, Eli},
  year={2017},
  publisher={Cambridge university press}
}

@misc{nakamoto2008bitcoin,
  title        = {Bitcoin: A peer-to-peer electronic cash system},
  author       = {Nakamoto, Satoshi},
  year         = {2008},
  howpublished = {\url{https://bitcoin.org/bitcoin.pdf}}
}

@unpublished{poon2016bitcoin,
  title  = {The Bitcoin Lightning Network: Scalable Off-Chain Instant Payments},
  author = {Poon, Joseph and Dryja, Thaddeus},
  year   = {2016},
  note   = {White paper},
  url    = {https://lightning.network/lightning-network-paper.pdf}
}

@inproceedings{OUSPS,
  TITLE = {{Optimal Uniform Shortest Path Sampling}},
  AUTHOR = {Dreyer, Simon and Genitrini, Antoine and Naima, Mehdi},
  URL = {https://hal.science/hal-04669060},
  BOOKTITLE = {{WALCOM: Algorithms and Computation}},
  ADDRESS = {Chengdu, China},
  PUBLISHER = {{Springer Nature Singapore}},
  SERIES = {Lecture Notes in Computer Science},
  VOLUME = {15411},
  PAGES = {160-179},
  YEAR = {2025},
  MONTH = Feb,
  DOI = {10.1007/978-981-96-2845-2\_11},
  HAL_ID = {hal-04669060},
  HAL_VERSION = {v3},
}

@article{girvan2002community,
  title={Community structure in social and biological networks},
  author={Girvan, Michelle and Newman, Mark EJ},
  journal={Proceedings of the national academy of sciences},
  volume={99},
  number={12},
  pages={7821--7826},
  year={2002},
  publisher={The National Academy of Sciences}
}

@inproceedings{lange2021impact,
  title={On the impact of attachment strategies for payment channel networks},
  author={Lange, Kimberly and Rohrer, Elias and Tschorsch, Florian},
  booktitle={2021 IEEE International Conference on Blockchain and Cryptocurrency (ICBC)},
  pages={1--9},
  year={2021},
  organization={IEEE}
}

@article{MBI,
author = {Bergamini, Elisabetta and Crescenzi, Pierluigi and D'Angelo, Gianlorenzo and Meyerhenke, Henning and Severini, Lorenzo and Velaj, Yllka},
title = {Improving the Betweenness Centrality of a Node by Adding Links},
year = {2018},
issue_date = {2018},
publisher = {Association for Computing Machinery},
address = {New York, NY, USA},
volume = {23},
issn = {1084-6654},
url = {https://doi.org/10.1145/3166071},
doi = {10.1145/3166071},
journal = {ACM J. Exp. Algorithmics},
month = aug,
articleno = {1.5},
numpages = {32},
}

@misc{lnd_github,
  author       = {Lightning Labs},
  title        = {lnd: Lightning Network Daemon},
  year         = {2025},
  howpublished = {\url{https://github.com/lightningnetwork/lnd}},
}

@inproceedings{ersoy2020profit,
author = {Ersoy, O\u{g}uzhan and Roos, Stefanie and Erkin, Zekeriya},
title = {How to Profit from Payments Channels},
pages = {284–303},
numpages = {20},
booktitle = {Financial Cryptography and Data Security: 24th International Conference, FC 2020, Kota Kinabalu, Malaysia, February 10–14, 2020, Revised Selected Papers},
year = {2020},
publisher = {Springer-Verlag},
location = {Kota Kinabalu, Malaysia},
isbn = {978-3-030-51279-8},
address = {Berlin, Heidelberg},
doi = {10.1007/978-3-030-51280-4_16},
numpages = {20},
}

@inproceedings{Dandekar2011liquidity,
author = {Dandekar, Pranav and Goel, Ashish and Govindan, Ramesh and Post, Ian},
title = {Liquidity in credit networks: a little trust goes a long way},
year = {2011},
isbn = {9781450302616},
publisher = {Association for Computing Machinery},
address = {New York, NY, USA},
url = {https://doi.org/10.1145/1993574.1993597},
doi = {10.1145/1993574.1993597},
booktitle = {Proceedings of the 12th ACM Conference on Electronic Commerce},
pages = {147–156},
numpages = {10},
location = {San Jose, California, USA},
series = {EC '11}
}

@inproceedings{Ghosh2007Mechanism,
author = {Ghosh, Arpita and Mahdian, Mohammad and Reeves, Daniel M. and Pennock, David M. and Fugger, Ryan},
title = {Mechanism design on trust networks},
year = {2007},
isbn = {3540771042},
publisher = {Springer-Verlag},
address = {Berlin, Heidelberg},
booktitle = {Proceedings of the 3rd International Conference on Internet and Network Economics},
pages = {257–268},
numpages = {12},
location = {San Diego, CA, USA},
series = {WINE'07}
}

@INPROCEEDINGS{Defigueiredo2005trustdavis,
  author={DeFigueiredo, Dd.B. and Barr, E.T.},
  booktitle={Seventh IEEE International Conference on E-Commerce Technology (CEC'05)}, 
  title={TrustDavis: a non-exploitable online reputation system}, 
  year={2005},
  volume={},
  number={},
  pages={274-283},
  doi={10.1109/ICECT.2005.98}
}

@article{
Lain1978Lsquares,
author = {T. L. Lai  and Herbert Robbins  and C. Z. Wei },
title = {Strong consistency of least squares estimates in multiple regression},
journal = {Proceedings of the National Academy of Sciences},
volume = {75},
number = {7},
pages = {3034-3036},
year = {1978},
doi = {10.1073/pnas.75.7.3034},
URL = {https://www.pnas.org/doi/abs/10.1073/pnas.75.7.3034},
eprint = {https://www.pnas.org/doi/pdf/10.1073/pnas.75.7.3034},
}

@article{Rene2021Security,
  author       = {Rene Pickhardt and
                  Sergei Tikhomirov and
                  Alex Biryukov and
                  Mariusz Nowostawski},
  title        = {Security and Privacy of Lightning Network Payments with Uncertain
                  Channel Balances},
  journal      = {CoRR},
  volume       = {abs/2103.08576},
  year         = {2021},
  url          = {https://arxiv.org/abs/2103.08576},
  eprinttype    = {arXiv},
  eprint       = {2103.08576},
  bibsource    = {dblp computer science bibliography, https://dblp.org}
}

@article{Goel2020ContinuousCN,
  title={Continuous Credit Networks and Layer 2 Blockchains: Monotonicity and Sampling},
  author={Ashish Goel and Geoffrey Ramseyer},
  journal={Proceedings of the 21st ACM Conference on Economics and Computation},
  year={2020},
  url={https://api.semanticscholar.org/CorpusID:215754249}
}

@article{dotan2021survey,
  title={Survey on blockchain networking: Context, state-of-the-art, challenges},
  author={Dotan, Maya and Pignolet, Yvonne-Anne and Schmid, Stefan and Tochner, Saar and Zohar, Aviv},
  journal={ACM Computing Surveys (CSUR)},
  volume={54},
  number={5},
  pages={1--34},
  year={2021},
  publisher={ACM New York, NY, USA}
}
